\documentclass[runningheads]{llncs}

\usepackage{times}
\usepackage{amsmath}
\usepackage{amssymb}
\usepackage{diagbox}
\usepackage{url}
\usepackage{graphicx}

\newtheorem{observation}[theorem]{Observation}
 
\title{Fair Sharing: The Shapley Value for Ride-Sharing and Routing Games}

 \author{Chaya Levinger\inst{1} \and
 Noam Hazon\inst{1} \and
 Amos Azaria\inst{1}}
 \authorrunning{Levinger, Hazon and Azaria}
 \institute{Department of Computer Science, Ariel University, Israel
 \email{chaya.levinger@msmail.ariel.ac.il}\\
 \email{\{noamh,amos.azaria\}@ariel.ac.il}}

\begin{document}

\maketitle

\begin{abstract}
Ride-sharing services are gaining popularity and are crucial for a sustainable environment. A special case in which such services are most applicable, is the last mile variant. In this variant it is assumed that all the passengers are positioned at the same origin location (e.g. an airport), and each have a different destination. 
One of the major issues in a shared ride is fairly splitting of the ride cost among the passengers.

In this paper we use the Shapley value, which is one of the most significant solution concepts in cooperative game theory, for fairly splitting the cost of a shared ride.
We consider two scenarios. In the first scenario there exists a fixed priority order in which the passengers are dropped-off (e.g. elderly, injured etc.), and we show a method for efficient computation of the Shapley value in this setting. Our results are also applicable for efficient computation of the Shapley value in routing games. 

In the second scenario there is no predetermined priority order. We show that the Shapley value cannot be efficiently computed in this setting. However, extensive simulations reveal that our approach for the first scenario can serve as an excellent proxy for the second scenario, outperforming other known proxies.  

\end{abstract}

\section{Introduction}

On-demand ride-sharing services, which group together passengers with similar itineraries, can be of significant social and environmental benefit, by reducing travel costs, road congestion and $CO_2$ emissions. 
Indeed, the National Household Travel Survey performed in the U.S. in 2009 \cite{santos2011summary} revealed that approximately 83.4\% of all trips in the U.S. were in a private vehicle (other options being public transportation, walking, etc.). The average vehicle occupancy was only $1.67$ when compensating for the number of passengers. 
The deployment of autonomous cars in the near future is likely to increase the spread for ride-sharing services, since it will be easier and cheaper for a company to handle a fleet of autonomous cars that can serve the demands of different passengers.

Most works in the domain of ride-sharing are dedicated to the assignment of passengers to vehicles, or to planning optimal drop-off routes~\cite{psaraftis2016dynamic,Alonso-Mora462,molenbruch2017}.
In this paper we study a fair allocation of the cost of the shared ride in the last mile variant~\cite{cheng2014mechanisms}. That is, we analyze the cost allocation when all passengers are positioned at the same origin location. 
We concentrate on the Shapley value~\cite{shapley1953value} as our notion of fair cost allocation. The Shapley value is widely used in cooperative games, and is the only cost allocation satisfying efficiency, symmetry, null player property and additivity. 
The Shapley value has been even termed the most important normative division scheme in cooperative game theory \cite{winter2002shapley}. However, the Shapley value depends on the travel cost of a ride of each subset of the passengers.
Therefore, as stated by {\"O}zener and Ergun ~\cite{ozener2008allocating}, ``In general, explicitly calculating the Shapley value requires exponential time. Hence, it is an impractical cost-allocation method unless an implicit technique given the particular structure of the game can be found''.

There are two possible general structures of the last-mile ride-sharing problem. In some cases there is a priority order in which the passengers are dropped-off.
Such prioritization may be attributed to the order in which the passengers arrived at the origin location, or the frequency of passenger usage of the service; the latter is similar to the different boarding groups on an aircraft. Other rationales for prioritization may include urgency of arrival or priority groups in need (e.g. elderly, disabled, pregnant women, and the injured). Clearly, in such cases, the prioritization is preserved when determining the travel cost of a ride with a subset of the passengers.
We denote this problem as the \emph{prioritized ride-sharing problem}.
Indeed, in some scenarios there is no predetermined prioritization order. In such cases it is assumed that a ride with a subset of the passengers is performed using the shortest (or cheapest) path that traverses their destinations. We denote this problem as the \emph{non-prioritized ride-sharing problem}.

Our problems of cost allocation are closely related to traveling salesman games~\cite{potters1992traveling}. In these games, a service provider makes a round-trip along the locations of several sponsors, where the total cost of the trip should be distributed among the sponsors. 
Specifically, the prioritized ride-sharing problem is similar to the fixed-route traveling salesman game, also known as routing game~\cite{yengin2012characterizing}, while the non-prioritized ride-sharing problem is similar to the traveling salesman game.  
Most of the works on traveling salesman games concentrated on finding an element of the core, a solution game concept which is different from the Shapley value. One exception is the work of Yengin~\cite{yengin2012characterizing}, who has studied the Shapley value of routing games and has conjectured that there is no efficient way for computing the Shapley value in routing games. 

In this paper, we show an efficient computation of the Shapley value for the prioritized ride-sharing problem. Our method is based on smart enumeration of the components that are used in the computation of the Shapley value. Furthermore, our approach can be generalized to routing games, and we thus also provide an efficient way for computing the Shapley value in routing game. We then move to analyze the non-prioritized ride-sharing problem and show that, unless P=NP, there is no polynomial time algorithm for computing the Shapley value. Fortunately, we show through extensive simulations that when the given travel path is the shortest path the Shapley value of the prioritized ride-sharing problem can be used as an excellent proxy for the Shapley value of the non-prioritized ride-sharing problem.


We note that the term ride-sharing is used in the literature with different meanings. We consider only the setting where the vehicle operator does not have any preferences or predefined destination. Instead, the vehicle's route is determined solely by the passengers' requests.
In addition, the context of our work is that the assignment of the passengers to the vehicle has already been determined, either by a ride-sharing system or by the passengers themselves, and we only need to decide on the cost allocation. 
Since we focus on the case where the assignment has already been determined, we do not consider the ability of passengers to deviate from the given assignment and join a different vehicle, which is acceptable since either they want to travel together or no other alternative exists. 

To summarize, the contributions of this paper are two-fold:
\begin{enumerate}
    \item We show an efficient method for computing the Shapley value of each user in a shared-ride when the priority order is predetermined. 
Our solution entails that the Shapley value can be computed in polynomial time in routing games as well, which is in contrast to a previous conjecture made.
\item We show that, while there exist no polynomial algorithm for computing the Shapley value of the non-prioritized ride-sharing problem (unless P=NP), the Shapley value of the prioritized ride-sharing problem can be used as an excellent proxy for the Shapley value of the non-prioritized ride-sharing problem (under the assumption that the provided travel path is the shortest path).
\end{enumerate}

\section{Related Work}
The ride-sharing cost allocation problems that we study are closely related to traveling salesman games~\cite{potters1992traveling}. 
Specifically, the prioritized ride-sharing problem is similar to the fixed-route traveling salesman game~\cite{fishburn1983fixed,potters1992traveling,besozzi2014traveling}, also known as routing game~\cite{yengin2012characterizing}.

One variant of routing game is the fixed-route traveling salesman problems with appointments. In this variant the service provider is assumed to travel back home (to the origin) when she skips a sponsor. This variant was introduced by Yengin~\cite{yengin2012characterizing}, who also showed how to efficiently compute the Shapley value for this problem but stated that his technique does not carry over to routing games.

The prioritized ride-sharing problem can also be interpreted as a generalization of the airport problem \cite{littlechild1973simple} to a two dimensional plane.
In the airport problem we need to decide how to distribute the cost of an airport runway among different airlines who need runways of different lengths.
In our case we distribute the cost among passengers who need rides of different lengths and destinations.
It was shown that the Shapley value can be efficiently computed for the airport problem, however achieving efficient computation of the Shapley value in our setting requires a different technique.

The Shapley value for the traveling salesman game, which is related to our non-prioritized ride-sharing problem, has rarely received serious attention in the literature, due to its computational complexity. Notably, Aziz et al.~\cite{aziz2016study} suggested a number of direct and sampling-based procedures for calculating the Shapley value for the traveling salesman game. They further surveyed several proxies for the Shapley value that are relatively easy to compute, and experimentally evaluate their performance. We develop a proxy for the Shapley value for the non-prioritized ride-sharing problem which is based on the Shapley-value for the prioritized ride-sharing problem, and compare its performance with proxies that are based on the work of Aziz et al.

The problem of fair cost allocation was also studied in the context of logistic operation. In this domain, shippers collaborate and bundle their shipment requests together to achieve better rates from a carrier \cite{guajardo2016review}. The Shapley value was also investigated in this domain of collaborative transportation~\cite{frisk2010cost,sun2015transportation}.
In particular, {\"O}zener and Ergun~\cite{ozener2008allocating} stated that ``we do not know of an efficient technique for calculating the Shapley value for the shippers' collaboration game''. Indeed, Fiestras-Janeiro et al.~\cite{Fiestras-Janeiro2012} developed the line rule, which is inspired by the Shapley value, but requires less computational effort and relates better with the core. However, the line rule is suitable for a specific inventory transportation problem. 
{\"O}zener~\cite{ozener2014developing} describes an approximation of the Shapley Value when trying to simultaneously allocate both the transportation costs and the emissions among the customers. Overall, we note that the main requirements from a cost allocation in the context of logistic operation is stability, and an equal distribution of the profit, since the collaboration is assumed to be long-termed. The type of interaction is our setting is inherently different, as it is an ad-hoc short term collaboration.

%
%

In another domain, Bistaffa et al.~\cite{bistaffa2015recommending} introduce a fair payment scheme, which is based on the game theoretic concept of the kernel, for the social ride-sharing problem (where the set of commuters are connected through a social network).

\section{Preliminaries}
We are given a weighted graph $G(V,E)$ that represents a road network; $V$ is the set of possible locations, and $E$ is a set of weighted edges that represents the set of roads. 
We are given a set $U = \{u_1, u_2, ..., u_n\}$ of passengers (users) that depart from the same origin, $d_0 \in V$. Without loss of generality, we assume that passenger $u_1$ will be dropped-off first, passenger $u_2$ will be dropped off second, etc. Each passenger $u_i$ has a corresponding destination $d_i \in V$. Let $D \subset V$ be the set of destinations, $D = \{d_1, d_2, ..., d_n\}$. 
We denote by $\delta(d_i,d_j)$ the shortest travel distance between $d_i$ and $d_j$ in $G$ and $\delta(d_i,d_i)=0$.
To simplify the notation we define a dummy destination, $d_{n+1}$, such that for every $i \in \{0,1,...,n\}$, $\delta(d_i,d_{n+1})=0$.
Given a set $S \subseteq D$, let $c(S)$ be the cost associated with the subset $S$. That is, $c(D)$ is the total travel cost of the shared ride. We note that $c(S)$, where $S \subsetneq D$, depends on the order in which the passengers are dropped off, and therefore $c(S)$ is defined differently in the prioritized ride-sharing problem and in the non-prioritized ride-sharing problem. The Shapley value for a passenger $u_i$ is formally defined as:
\[
\phi(u_i)=\sum_{S \subseteq U \setminus \{i\}} \frac{|S|!(|U|-|S|-1)!}{|U|!} \big(c(S \cup \{i\}) - c(S)\big).
\]
That is, the Shapley value is an average over the marginal costs of each passenger. 

\section{The Prioritized Ride-sharing Problem}
\label{sec:shapley}
In this section we assume that the passengers are ordered according to some predetermined priority order, and efficiently compute the payment for every passenger using the Shapley value. 
Unlike other related work \cite{potters1992traveling}, we do not require that the priority order will be the optimal order that minimizes the total cost.

\subsection{Notation}
Given the set of passengers $U$, without loss of generality, we assume that passenger $u_1$ has the highest priority, passenger $u_2$ has the second highest priority, etc. 
%
%
%
Given a set $S \subseteq D$, let $\tilde{S}$ be the set $S$ ordered in an ascending order (according to the priority order), and let $S[i]$ be the destination that is in the $i$-th position in $\tilde{S}$. For ease of notation we use $S[0]$ to denote $d_0$ and $S[|S|+1]$ to denote $d_{n+1}$.

Given a set $S \subseteq D$, let $v(S)$ be the shortest travel distance of the path that starts at the origin $d_0$ and traverses all destinations ${d_i} \in S$ according to an ascending order. That is, $v(S) = \sum\limits_{i=0}^{|S|-1} \delta(S[i],S[i+1])$. This value ($v(S)$) serves as the cost associated with a subset of passengers, $c(S)$, in the computation of the Shapley value.

Let $R$ be a permutation on $D$ and let $P^{R}_{i}$ be the set of the previous destinations to $d_i$ in permutation $R$.

\subsection{Efficient Computation of the Shapley Value}
We are interested in determining the payment for each passenger, $u_i$, according to the Shapely value, $\phi(u_i)$. The Shapley value has several equivalent formulas, and we use the following formula to derive an efficient computation in the prioritized ride-sharing problem: 
\[
\phi(u_i)={\frac {1}{n!}}\sum _{R}\bigg(v(P_{i}^{R}\cup \left\{d_i\right\})-v(P_{i}^{R})\bigg).
\]

Given a permutation $R$ and a passenger $u_i$, let $d_l \in P_{i}^{R}$ be a destination such that $l < i$ and $\forall d_j \in  P_{i}^{R}, j \leq l$ or $i < j$.
If no such destination exists, then $d_l$ is defined as $d_0$.
Similarly, let $d_r \in P_{i}^{R}$ be a destination such that $i < r$ and $\forall d_j \in  P_{i}^{R}, j < i$ or $r \leq j$.
If no such destination exists, then $d_r$ is defined as $d_{n+1}$.
We use $\ell$ (and \textit{r}) to denote the position of $d_l$ (and $d_r$) in the ordered $\tilde{P_{i}^{R}}$, respectively. If $d_l = d_0$ then $\ell=0$, and if $d_r=d_{n+1}$ then \textit{r}$=|P^{R}_{i}|+1$. We note that $P_{i}^{R}[\ell] = d_l$, $P_{i}^{R}[\textit{r}] = d_r$ and $\textit{r} = \ell+1$.

For example, assume $D = \{d_1, d_2, d_3, d_4, d_5, d_6\}$ and $R= \{d_6, d_2, d_5, d_4, d_3, d_1\}$, we get $P_{4}^{R} = \{d_6, d_2, d_5\}$ and thus $\tilde{P_{4}^{R}} = \{d_2, d_5, d_6\}$, $d_l = d_2$ (i.e., $\ell = 1$), $d_r=d_5$ (i.e., $\textit{r}=2$), and $P_{4}^{R}[\ell] = d_2$.


Our first observation is that the computation of the Shapley value in our setting, $\phi(u_i)$, may be written as the sum over the distances between pairs of destinations.
\begin{observation}
\label{obs:shap_sums}
$\phi(u_i) = \frac{1}{n!} \sum\limits_{p=0}^{n-1}\sum\limits_{q=p+1}^{n} \alpha_{p,q}^i\delta(d_p,d_q)$, for some $\alpha_{p,q}^i \in \mathbb{Z}$.
\end{observation}
\begin{proof}
We note that $\phi(u_i) \cdot n!$ is a sum over $v(S)$ for multiple $S \subseteq D$. By definition, $v(S) = \sum\limits_{j=0}^{|S|-1} \delta(S[j],S[j+1])$, such that $S[j] = d_p$ and $S[j+1] = d_q$ where $p < q$.
\qed
\end{proof}

We now show that we can rewrite the computation of the Shapely value in our setting as follows.
\begin{lemma}
\label{lem:shap_to_sum}
\[    \phi(u_i)={\frac {1}{n!}}\sum\limits _{R}\bigg( \delta(d_l, d_i) + \delta(d_i, d_r) - \delta(d_l, d_r)\bigg)
\]
\end{lemma}
\begin{proof}

\[v(P^{R}_{i})=\sum\limits_{j=0}^{|P^{R}_{i}|-1}\delta(P^{R}_{i}[j],P^{R}_{i}[j+1])=\] 
\[\sum\limits_{j=0}^{\ell-1}\delta(P^{R}_{i}[j],P^{R}_{i}[j+1])+\delta(d_l,d_r)+\sum\limits_{j=\textit{r}}^{|P^{R}_{i}|-1}\delta(P^{R}_{i}[j],P^{R}_{i}[j+1])\]

In addition, \[v(P^{R}_{i}\cup \left\{d_i\right\})=\sum\limits_{j=0}^{\ell-1}\delta(P^{R}_{i}[j],P^{R}_{i}[j+1])+\]\[\delta(d_l,d_i)+\delta(d_i,d_r)+\sum\limits_{j=\textit{r}}^{|P^{R}_{i}|-1}\delta(P^{R}_{i}[j],P^{R}_{i}[j+1]).\]

By definition, \[ \phi(u_i)={\frac {1}{n!}}\sum _{R}\left[v(P^{R}_{i}\cup \left\{d_i\right\})-v(P^{R}_{i})\right]=\]

${\frac {1}{n!}}\sum\limits _{R}\bigg( \sum\limits_{j=0}^{\ell-1}\delta(P^{R}_{i}[j], P^{R}_{i}[j+1]) + \delta(d_l, d_i) + \delta(d_i, d_r) + \sum\limits_{j=\textit{r}}^{|P^{R}_{i}|-1}\delta(P^{R}_{i}[j], P^{R}_{i}[j+1]) - \Big(\sum\limits_{j=0}^{\ell-1}\delta(P^{R}_{i}[j], P^{R}_{i}[j+1]) + \delta(d_l, d_r) + \sum\limits_{j=\textit{r}}^{|P^{R}_{i}|-1}\delta(P^{R}_{i}[j], P^{R}_{i}[j+1])\Big) \bigg) =$

${\frac {1}{n!}}\sum\limits _{R}\Big(\delta(d_l, d_i) + \delta(d_i, d_r) - \delta(d_l, d_r)\Big)$
\qed
\end{proof}

Following Observation~\ref{obs:shap_sums} and Lemma~\ref{lem:shap_to_sum} we now show that we can rewrite the computation of the Shapely value as a sum over distances, that can be computed in polynomial time.
\begin{theorem}
\label{thr:poly}
For each $i$, 
$\phi(u_i) = \sum\limits_{p=0}^{i}\sum\limits_{q=i}^{n} \beta _{p,q}^i\delta(d_p,d_q)$, 
where $q\neq p$, and $\beta_{p,q}^i \in 
\mathbb{Q}$ are computed in polynomial time.
\end{theorem}
\begin{proof}
By definition, $l < i < r$.
According to Lemma~\ref{lem:shap_to_sum}  $\phi(u_i) \cdot n!$ is a sum over $\delta(d_p,d_q)$, where $p \leq i \leq q$.
There are several terms in this sum:

\begin{itemize}

\item $\beta_{0,i}^i$ multiplies $\delta(d_0,d_i)$. Now, 
$\delta(d_0,d_i)$ appears in $\phi(u_i)$ in every permutation $R$ when $d_l=d_0$. That is, in all of the permutations where destination $d_i$ appears before any other destination $d_x$ such that $x<i$. We now count the number of such permutations. There are $\binom{n}{i}$ options to place the destinations $d_1,d_2,...,d_i$ among the $n$ available destinations. For each such option there are $(i-1)!$ options to order the destinations $d_1,d_2,...,d_i$ such that $d_i$ is the first destination among them. Finally, there are $(n-i)!$ options to order the destinations $d_{i+1},d_{i+2},...,d_n$. Therefore, $\delta(d_0,d_i)$ appears in $\binom{n}{i} \cdot (i-1)! \cdot (n-i)! = \frac{n!}{i}$ permutations, and by inserting $\frac{1}{n!}$ into the sum we get that $\beta_{0,i}^i = \frac{1}{i}$.


\item For each $q>i$, $\beta_{0,q}^i$ multiplies $\delta(d_0,d_q)$. Now, 
$\delta(d_0,d_q)$ appears negatively in $\phi(u_i)$ in every permutation $R$ when $d_l=d_0$ and $d_r=d_q$. That is, in all of the permutations where destination $d_q$ appears before $d_i$ (i.e., $d_q \in P^R_i$), but any other destination $d_x$ such that $x<q$, appears after $d_i$. We now count the number of such permutations. There are $\binom{n}{q}$ options to place the destinations $d_1,d_2,...,d_i,...,d_q$ among the $n$ available destinations. For each such option there are $(q-2)!$ options to order the destinations $d_1,d_2,...,d_i$ such that $d_q$ is the first destination and $d_i$ is the second destination among them. Finally, there are $(n-q)!$ options to order the destinations $d_{q+1},d_{q+2},...,d_n$. Therefore, $\delta(d_0,d_q)$ appears negatively in $\binom{n}{q} \cdot (q-2)! \cdot (n-q)! = \frac{n!}{q\cdot (q-1)}$ permutations, and by inserting $\frac{1}{n!}$ into the sum we get that $\beta_{0,q}^i = -\frac{1}{q\cdot (q-1)}$.


\item For each $0<p<i$, $\beta_{p,i}^i$ multiplies $\delta(d_p,d_i)$. Now, 
$\delta(d_p,d_i)$ appears in $\phi(u_i)$ in every permutation $R$ when $d_l=d_p$. That is, in all of the permutations where destination $d_p$ appears before $d_i$ (i.e., $d_p \in P^R_i$), but any other destination $d_x$ such that $p<x<i$, appears after $d_i$. We now count the number of such permutations. There are $\binom{n}{i-p+1}$ options to place the destinations $d_p,d_{p+1},...,d_i$ among the $n$ available destinations. For each such option there are $(i-p+1-2)!$ options to order the destinations $d_p,d_{p+1},...,d_i$ such that $d_p$ is the first destination and $d_i$ is the second destination among them. Finally, there are $(n-(i-p+1))!$ options to order the destinations $d_1, d_2,...,d_{p-1},d_{i+1},d_{i+2},...,d_n$. Therefore, $\delta(d_p,d_i)$ appears in $\binom{n}{i-p+1} \cdot (i-p-1)! \cdot (n-(i-p+1))! = \frac{n!}{(i-p)\cdot (i-p+1)}$ permutations, and by inserting $\frac{1}{n!}$ into the sum we get that $\beta_{p,i}^i = \frac{1}{(i-p)\cdot (i-p+1)}$.


\item For each $q>i$, $\beta_{i,q}^i$ multiplies $\delta(d_i,d_q)$. Now, 
$\delta(d_i,d_q)$ appears in $\phi(u_i)$ in every permutation $R$ when $d_r=d_q$. That is, in all of the permutations where destination $d_q$ appears before $d_i$ (i.e., $d_q \in P^R_i$), but any other destination $d_x$ such that $i<x<q$, appears after $d_i$. We now count the number of such permutations. There are $\binom{n}{q-i+1}$ options to place the destinations $d_i,d_{i+1},...,d_q$ among the $n$ available destinations. For each such option there are $(q-i+1-2)!$ options to order the destinations $d_i,d_{i+1},...,d_q$ such that $d_q$ is the first destination and $d_i$ is the second destination among them. Finally, there are $(n-(q-i+1))!$ options to order the destinations $d_1, d_2,...,d_{i-1},d_{q+1},d_{q+2},...,d_n$. Therefore, $\delta(d_p,d_i)$ appears in $\binom{n}{q-i+1} \cdot (q-i-1)! \cdot (n-(q-i+1))! = \frac{n!}{(q-i)\cdot (q-i+1)}$ permutations, and by inserting $\frac{1}{n!}$ into the sum we get that $\beta_{i,q}^i = \frac{1}{(q-i)\cdot (q-i+1)}$.


\item For each $p,q$ such that $p<i<q$, $\beta_{p,q}^i$ multiplies $\delta(d_p,d_q)$. Now, 
$\delta(d_p,d_q)$ appears negatively in $\phi(u_i)$ in every permutation $R$ when $d_l=d_p$ and $d_r=d_q$. That is, in all of the permutations where destinations $d_p, d_q$ appear before $d_i$ (i.e., $d_p,d_q \in P^R_i$), but any other destination $d_x$ such that $p<x<q, x\neq i$, appears after $d_i$. We now count the number of such permutations. There are $\binom{n}{q-p+1}$ options to place the destinations $d_p,d_{p+1},...,d_i,...,d_q$ among the $n$ available destinations. For each such option there are $(q-p+1-3)!$ options to order the destinations $d_p,d_{p+1},...,d_i,...,d_q$ such that $d_p$ is the first destination, $d_q$ is the second and $d_i$ is the third destination among them. Similarly, there are $(q-p+1-3)!$ options to order these destinations such that $d_q$ is the first destination, $d_p$ is the second and $d_i$ is the third. Finally, there are $(n-(q-p+1))!$ options to order the destinations $d_1, d_2,...,d_{p-1},d_{q+1},d_{q+2},...,d_n$. Therefore, $\delta(d_p,d_q)$ appears in $\binom{n}{q-p+1} \cdot 2 \cdot (q-p-2)! \cdot (n-(q-p+1))! =\frac{2 \cdot n!}{(q-p-1)\cdot (q-p) \cdot (q-p+1)}$ permutations, and by inserting $\frac{1}{n!}$ into the sum we get that $\beta_{p,q}^i = -\frac{2}{(q-p-1)\cdot (q-p) \cdot (q-p+1)}$.
%
\end{itemize}
\qed
\end{proof}

We note that the prioritized ride-sharing problem is very similar to the setting of routing games~\cite{potters1992traveling}. The model of routing games is of one service provider that makes a round-trip along the locations of several sponsors in a fixed order, where the total cost of the trip should be distributed among the sponsors. Clearly, our problem is almost identical: the service provider corresponds to the vehicle and the sponsors correspond to the passengers. The only difference is that in a routing game the sponsors also pay the cost of the trip back to the origin. Indeed, the results presented in this section carry over to routing games.
\begin{theorem}
The Shapley value in routing games can be computed in polynomial time.
\end{theorem}
\begin{proof}[Proof (sketch)]
We use our previous definitions and results with the following slight modifications. The dummy destination $d_{n+1}$ becomes $d_0$. Thus, $\delta(d_i,d_{n+1})=\delta(d_i,d_0)$. In Observation~\ref{obs:shap_sums} we need to modify the bound in the outer sum (with the index $p$) to $n$ and the bound in the inner sum (with the index $q$) to $n+1$. In addition, we use the proof of  Theorem~\ref{thr:poly}, but we add $\sum\limits_{p=0}^{i} \beta_{p,n+1}^i\delta(d_p,d_{n+1})$ to the calculation of $\phi(u_i)$, where for $p<i$, $\beta_{p,n+1}^i = -\frac{1}{(n-p)\cdot (n-p+1)}$ and $\beta_{i,n+1}^i = \frac{1}{n-i+1}$.
\qed
\end{proof}
Note that this is an unexpected result, since it refutes the conjecture in~\cite{yengin2012characterizing} that there is no efficient way for computing the Shapley value in routing games.

\section{Non-prioritized Ride-sharing Problem}
Similar to the prioritized ride-sharing problem we are given an initial priority order, which determines the drop-off order of the passengers. However, in the non-prioritized variant we do not enforce the fixed order for every subset of passengers. Instead, given a strict subset of passengers $S$, the cost associated with it, $c(S)$, is the length of the shortest path that traverses all of the destinations of the passengers in $S$. 

\subsection{The Hardness of the Non-prioritized Ride-sharing Problem}
In Section~\ref{sec:shapley} we showed that we can efficiently compute the Shapley value for the prioritized ride-sharing problem. In essence, the computation could be done efficiently since most of the travel distances cancel out, and only a polynomial number of terms remain in the computation. Unfortunately, this is not the case with the non-prioritized ride-sharing problem, where the Shapley value cannot be computed efficiently unless $P=NP$. 

Clearly, finding the length of the shortest path (not necessarily a simple path) that starts at a specific node, $v_0$, and traverses all nodes in a graph (without returning to the origin) cannot be performed in polynomial time, unless $P=NP$. We denote this problem as \emph{path-TSP}.
%
We use the path-TSP to show that computing the Shapley value for the prioritized ride-sharing cannot be done efficiently, unless $P=NP$.
\begin{theorem}
There is no polynomial time algorithm that computes the Shapley value for a given passenger in the prioritized ride-sharing problem unless $P=NP$.
\end{theorem}
\begin{proof}
Given an instance of the path-TSP problem on a graph $G(V,E)$ we denote the solution by $x$. We construct an instance of the non-prioritized ride-sharing problem as follows. We build a graph $G'(V',E')$, where we add a node $v'$, i.e., $V' = V \cup \{v'\}$. If $e \in E$ then $e \in E'$, and for all $v \in V$, $(v,v') \in E'$ with a weight of $M$, where $M$ is the sum of weights of all the edges in $E$. Finally, we set $|U|=|V|$, $D=V' \setminus \{v_0\}$, $d_0 = v_0$, and the drop-off order is arbitrarily chosen. Recall that $c(D)$ is the total travel cost associated with the chosen drop-off order. We ask to compute the Shapley value of the user $u'$ that is associated with the destination $v'$.

Clearly, the marginal contribution of $u'$ to any strict subset of $V$ is exactly $M$. However, the marginal contribution of $u'$ to the complete set $V$ is exactly $c(D)$ minus $x$ (the length of the shortest path starting at $v_0$ and traversing all nodes in $V$). That is,
\[
\phi(u') = \frac{(|U|-1)!}{|U|!}(c(D)-x) + \frac{|U|!-(|U|-1)!}{|U|!}M
\]
After some simple mathematical manipulations we get that $x = |U|\phi(u')-(|U|-1)M+c(D)$. Therefore, if we can compute $\phi(u')$ in polynomial time then we can solve the path-TSP problem in polynomial time, which is not possible unless $P=NP$.
\qed
\end{proof}

\begin{figure*}[hbpt]
\centering
\includegraphics[width=8cm]{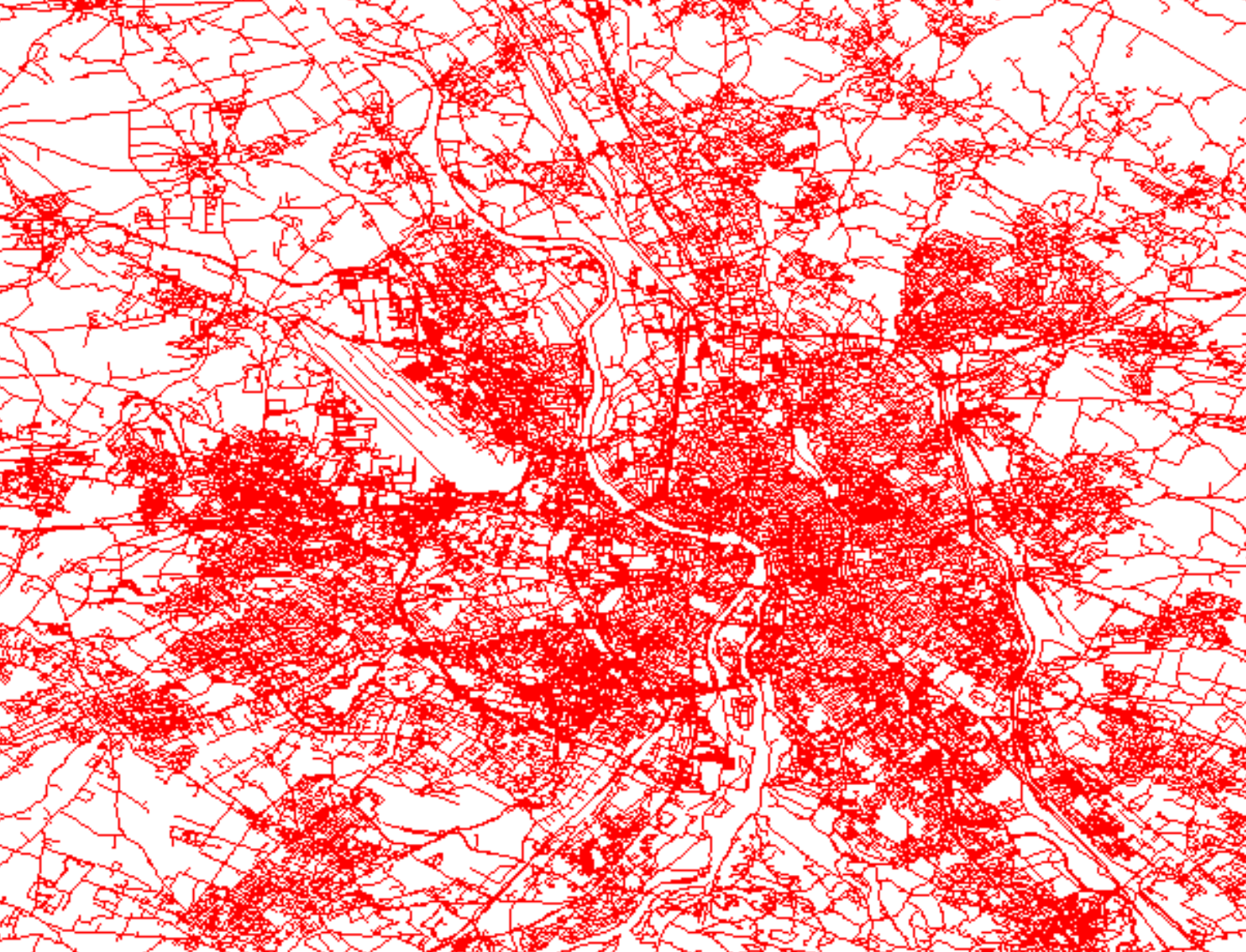} \caption{A graph created from a map of the city of Toulouse, France.}
\label{fig:ToulouseGraph}
\end{figure*}

\subsection{Shapley Approximation based on a Prioritized Order}

In Section \ref{sec:shapley} we presented a method for efficiently computing the Shapley value when a prioritization exists. In this section we show that our solution may be also applicable to the non-prioritized ride-sharing problem as an efficient proxy for the Shapley value. We term our proxy SHAPO: SHapley Approximation based on a Prioritized Order.

We compare SHAPO with the following three proxies for computing the Shapley value in traveling salesman games, that are in use in real-world applications~\cite{aziz2016study}.

\paragraph{Depot Distance}
This method divides the total ride cost proportionally to the distance from the depot, i.e. $Depot(u_i)=\frac{\delta(d_0,d_i)}{\sum_{j=1}^{n}{\delta(d_0,d_j)}}c(D)$. For example, a passenger traveling to a destination that is twice as distant from the origin as another passenger has to pay twice the cost, regardless of the actual travel path.
We note that this method has outperformed all other methods in ~\cite{aziz2016study} on real data. 

\paragraph{Shortcut Distance}
This method divides the total cost proportional to the change realized by skipping a destination when traversing the given path. Formally, let $Cut_i =\delta(d_{i-1}, d_i) + \delta(d_i, d_{i+1}) - \delta(d_{i-1}, d_{i+1})$. Then, $Shortcut(u_i)=\frac{Cut_i}{\sum_{j=1}^{n}Cut_j}c(D)$. 

\paragraph{Re-routed Margin}
This method is a more sophisticated realization of the shortcut distance method. That is, instead of using the given path when skipping a destination, we compute the optimal path. Formally, $Reroute(u_i)=\frac{c(D)-c(D \setminus \{d_i\})}{\sum_{j=1}^{n}c(D)-c(D \setminus \{d_j\})}c(D)$.
Note that when evaluating this proxy we need to solve $n$ TSPs, one for leaving out each destination. This is the only proxy we consider that requires a non-negligible time to compute.


\subsubsection{Experimental Settings}
In order to evaluate the performance of SHAPO, we evaluated each of the methods for $3,4,5,6,7,8$ and $9$ passengers. For the road network we used the graph of the city of Toulouse, France\footnote{obtained from \url{https://www.geofabrik.de/data/shapefiles_toulouse.zip}} as presented in Figure \ref{fig:ToulouseGraph}. This graph includes the actual distances between the different vertices. To convert the distances to travel costs we assumed a cost of $\$1$ per kilometer. The graph also includes the Toulouse-Blagnac airport, which was set as the origin ($d_0$).
We cropped the graph to $40,000$ vertices, by running Dijkstra algorithm~\cite{Dijkstra1959} starting at the airport, sorting all vertices by their distance from the airport, and removing all farther away vertices (including those that are unreachable). The destination vertices were randomly sampled for every passenger using a uniform distribution over all vertices, and each of the methods was evaluated $100$ times against the true Shapley value of all passengers. 

For running the simulations we assume that the given order of the passengers is according to the shortest path. This is a reasonable assumption, since if there is no prioritization, it is very likely that, in order to reduce the overall cost, the vehicle would travel using the shortest path (computed once). We conjecture that the results presented in this paper will carry-out also to situations in which the given passenger order is very close to being optimal (but not necessarily the exact optimal order), but we leave it for future investigation.

\subsubsection{Results}
Figure~\ref{fig:RunTimeGraph} presents the running time, in seconds, required to compute the Shapley value and its proxies for all passengers on a single instance (in logarithmic scale). As expected, we can compute the proxies, except for the Re-routed margin proxy, almost instantaneously. 
However, due to the extensive time required to compute the Shapley value, and since we evaluate each method $100$ times, we only evaluate the performance of all methods with up-to 9 passengers.

\begin{figure*}[hbpt]
\centering
\includegraphics[width=8cm]{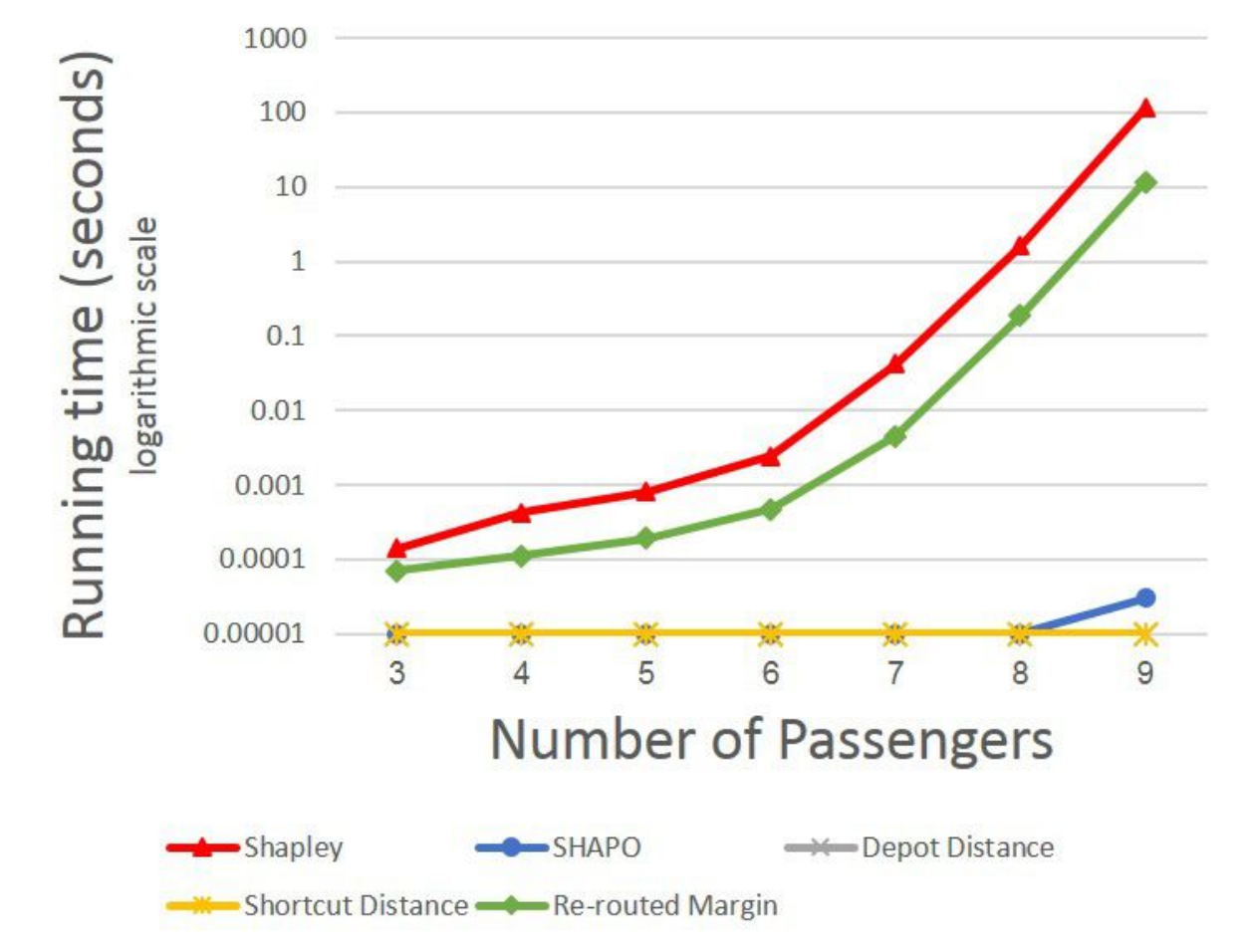} 
\caption{Running time, in seconds, required to compute a single instance of the Shapley value (in logarithmic-scale).}
\label{fig:RunTimeGraph}
\end{figure*}

We evaluate the performance of SHAPO against the three other proxies using $5$ different statistical measures (averaged on all $100$ iterations). We use $X(u_i)$ to denote the estimated Shapley value by the evaluated proxy.
\begin{enumerate}
    \item \textbf{Percent}: The average percentage of the deviation from the Shapley value. Formally, $Percent = \frac{1}{n}\sum_{i=1}^{n}\frac{|X(u_i) - \phi(u_i)|}{\phi(u_i)}$. 
    \item \textbf{MAE}: The mean absolute error, $MAE = \frac{1}{n}\sum_{i=1}^{n}|X(u_i) - \phi(u_i)|$. 
    \item \textbf{MSE}: The mean squared error, 
    $MSE = \frac{1}{n}\sum_{i=1}^{n}(X(u_i) - \phi(u_i))^2$. This measure gives higher weight to larger deviations.
    \item \textbf{RMSE}: The root mean squared error,  
    $RMSE = \sqrt{\frac{1}{n}\sum_{i=1}^{n}(X(u_i) - \phi(u_i))^2}$.
    \item \textbf{Max-Error}: The maximum deviation among all passengers between the real and estimated Shapley value,
    $Max = \max_{i=1}^{n}(|X(u_i) - \phi(u_i)|)$.
\end{enumerate}

\bgroup
\setlength{\tabcolsep}{0.7em}
\begin{table}
\centering
\begin{tabular}{ c|c c c c c c c c } 
\hline
 & 3 & 4 & 5 & 6 & 7 & 8 & 9 & AVG\\
\hline
SHAPO & \textbf{1.84\%} & \textbf{3.44\%} & \textbf{4.49\%} & \textbf{4.68\%} & \textbf{5.27\%} & \textbf{6.53\%} & \textbf{5.94\%} & \textbf{4.60\%}\\
Depot Distance & 20.02\% & 26.36\% & 29.52\% & 35.28\% & 35.78\% & 35.92\% & 35.91\% & 31.26\%\\
Shortcut Distance & 34.79\% & 41.55\% & 45.44\% & 53.70\% & 54.23\% & 53.35\% & 55.89\% & 48.42\%\\
Re-routed Margin & 84.59\% & 76.13\% & 79.65\% & 119.90\% & 80.57\% & 85.97\% & 72.83\% & 85.66\%\\
\hline
\end{tabular}
\caption{Average percentage of the deviation from the Shapley value (Percent). Averaged over $100$ iterations. Lower is better.}
\label{tab:percentTable}
\end{table}
\egroup

\bgroup
\setlength{\tabcolsep}{1em}
\begin{table}
\centering
\begin{tabular}{ c|c c c c c c c c } 
\hline
 & 3 & 4 & 5 & 6 & 7 & 8 & 9 & AVG\\
\hline
SHAPO & \textbf{\$0.11} & \textbf{\$0.18} & \textbf{\$0.20} & \textbf{\$0.20} & \textbf{\$0.20} & \textbf{\$0.22} & \textbf{\$0.21} & \textbf{\$0.19}\\
Depot Distance & \$1.11 & \$1.29 & \$1.34 & \$1.53 & \$1.38 & \$1.35 & \$1.28 & \$1.33\\
Shortcut Distance & \$1.87 & \$2.19 & \$2.23 & \$2.47 & \$2.27 & \$2.07 & \$2.09 & \$2.17\\
Re-routed Margin & \$3.21 & \$2.81 & \$2.45 & \$2.43 & \$2.10 & \$2.03 & \$1.85 & \$2.41\\
 \hline
\end{tabular}
\caption{The mean absolute error of the deviation from the Shapley value (MAE). Averaged over $100$ iterations. Lower is better.}
\label{tab:ABSTable}
\end{table}

\begin{table}
\centering
\begin{tabular}{ c|c c c c c c c c } 
\hline
 & 3 & 4 & 5 & 6 & 7 & 8 & 9 & AVG\\
 \hline
SHAPO & \textbf{0.068} & \textbf{0.115} & \textbf{0.124} & \textbf{0.098} & \textbf{0.102} & \textbf{0.159} & \textbf{0.117} & \textbf{0.112}\\
Depot Distance & 1.988 & 2.634 & 3.054 & 4.130 & 3.327 & 3.366 & 3.066 & 3.081\\
Shortcut Distance & 5.805 & 8.431 & 8.937 & 10.840 & 9.749 & 8.097 & 7.731 & 8.513\\
Re-routed Margin & 15.951 & 12.404 & 9.587 & 10.174 & 7.399 & 7.168 & 5.986 & 9.810\\
 \hline
\end{tabular}
\caption{The mean squared error of the deviation from the Shapley value (MSE). Averaged over $100$ iterations. Lower is better.}
\label{tab:PowTable}
\end{table}

\begin{table}
\centering
\begin{tabular}{ c|c c c c c c c c } 
\hline
 & 3 & 4 & 5 & 6 & 7 & 8 & 9 & AVG\\
 \hline
SHAPO & \textbf{0.121} & \textbf{0.204} & \textbf{0.239} & \textbf{0.246} & \textbf{0.259} & \textbf{0.299} & \textbf{0.277} & \textbf{0.235}\\
Depot Distance & 1.205 & 1.474 & 1.594 & 1.859 & 1.699 & 1.703 & 1.617 & 1.593\\
Shortcut Distance & 2.112 & 2.556 & 2.698 & 3.068 & 2.904 & 2.675 & 2.662 & 2.668\\
Re-routed Margin & 3.540 & 3.200 & 2.866 & 2.951 & 2.517 & 2.508 & 2.292 & 2.839\\
 \hline
\end{tabular}
\caption{The root mean squared error of the deviation from the Shapley value (RMSE). Averaged over $100$ iterations. Lower is better.}
\label{tab:RMSETable}
\end{table}

\begin{table}
\centering
\begin{tabular}{ c|c c c c c c c c } 
\hline
 & 3 & 4 & 5 & 6 & 7 & 8 & 9 & AVG\\
 \hline
SHAPO & \textbf{\$0.17} & \textbf{\$0.30} & \textbf{\$0.39} & \textbf{\$0.44} & \textbf{\$0.49} & \textbf{\$0.62} & \textbf{\$0.58} & \textbf{\$0.43}\\
Depot Distance & \$1.66 & \$2.24 & \$2.73 & \$3.44 & \$3.32 & \$3.51 & \$3.47 & \$2.91\\
Shortcut Distance & \$2.81 & \$3.92 & \$4.69 & \$5.66 & \$5.77 & \$5.49 & \$5.56 & \$4.84\\
Re-routed Margin & \$4.81 & \$4.74 & \$4.60 & \$5.32 & \$4.67 & \$4.95 & \$4.71 & \$4.83\\
 \hline
\end{tabular}
\caption{The maximum deviation among all passengers between the real and estimated Shapley value (Max-Error). Averaged over $100$ iterations. Lower is better.}
\label{tab:MAXTable}
\end{table}
\egroup

The results are depicted in Tables \ref{tab:percentTable}, \ref{tab:ABSTable}, \ref{tab:PowTable}, \ref{tab:RMSETable} and \ref{tab:MAXTable}. 
SHAPO significantly outperforms the other proxies in all measures, with any number of passengers evaluated. Despite the depot distance method outperforming the other two methods, SHAPO is between $5.5$ to $42.3$ times better than the depot distance in all measures. Note that the units of MAE and Max-Error are dollars. That is, as depicted in Table~\ref{tab:ABSTable}, SHAPO deviated by only $19$ cents, on average, from the actual Shapley value. The depot distance deviated by $\$1.33$, while the averaged shared-ride cost per passenger was approximately $\$5$. Similarly, the maximal deviation of SHAPO was less than $44$ cents (on average), while the maximal deviation of the depot distance was more than $\$2.9$.

\section{Conclusions}
The Shapley value is considered one of the most important division scheme of revenues or costs, but its direct computation is often not practical for a reasonable size game. Therefore, Mann and Shapely~\cite{mann1962values} suggest to consider restrictions and constraints in order to find games where the Shapley value can be efficiently computed. The airport problem~\cite{littlechild1973simple} is one example of these games, where the Shapley value can be efficiently computed. Our prioritized ride-sharing problem is a generalization of the airport problem to the 2D plane, and we showed that the Shapley can be efficiently computed in this generalization as well. However, we show that the non-prioritized ride-sharing problem, which is possibly the next level of generalization, cannot be efficiently computed (unless $P=NP$). Interestingly, the prioritized  ride-sharing can still serve as an efficient proxy for the Shapley value of the non-prioritized ride-sharing problem where the provided travel path is the shortest path.

There are several interesting directions for future work. One possible  direction is to compare our proxy for computing the Shapley value in the non-prioritized ride-sharing problem to a sampling based approach~\cite{castro2009polynomial}. It is expected that a sampling based approach will be more accurate if there is a sufficient number of samples, but it will certainly require a lot more computation time. It is thus interesting to analyze when our proxy is still better than a sampling-based approach, and when it is the point in which a sampling-based approach becomes better than our proxy.
From a theoretical perspective, we showed that computing the Shapley value for the non-prioritized ride-sharing problem is a hard problem. However, the hardness may be derived also from the hardness of path-TSP. There are several polynomial time approximation and heuristics for TSP that can be adjusted for path-TSP. It is thus interesting to analyze the computational complexity of finding the Shapley value, where $c(S)$ is computed using one of these approximations or heuristics.


\bibliographystyle{splncs03}
\bibliography{shapley}
\end{document}